\documentclass[journal]{IEEEtran}
\usepackage{color,xcolor,colortbl}
\usepackage{stfloats}
\usepackage{array}
\usepackage{arydshln}
\usepackage{booktabs}
\usepackage{threeparttable}
\usepackage{multirow}

\def\black{\color{black}}
\usepackage{graphicx}
\usepackage{hyperref}
\usepackage{tabularx}
\hypersetup{hidelinks,
	colorlinks=true,
	allcolors=black,
	pdfstartview=Fit,
	breaklinks=true
}
\usepackage{amsmath,bm}
\usepackage{amssymb}
\usepackage{amsthm}
\usepackage[ruled,linesnumbered]{algorithm2e}
\usepackage[normalem]{ulem}

\newtheoremstyle{IEEEtheorem}
  {}
  {}
  {}
  {10pt}
  {\itshape}
  {:}
  { }
  {\thmname{#1}\thmnumber{ #2}\thmnote{ (#3)}}

\theoremstyle{IEEEtheorem}

\newtheorem{thm}{Theorem}
\newtheorem{ass}{Assumption}
\newtheorem{cor}{Corollary}
\newtheorem{lem}{Lemma}
\newtheorem{rmk}{Remark}
\newtheorem{dfn}{Definition}
\newtheorem{prop}{Proposition}

\newtheorem{prob}{Problem}

\usepackage[capitalize, nameinlink]{cleveref}
\crefname{thm}{Theorem}{Theorems}
\crefname{prob}{Problem}{Problems}
\crefname{ass}{Assumption}{Assumptions}
\crefname{prop}{Proposition}{Propositions}
\crefname{lem}{Lemma}{Lemmas}
\crefname{rmk}{Remark}{Remarks}
\crefname{cor}{Corollary}{Corollaries}
\crefname{dfn}{Definition}{Definitions}

\usepackage{float}
\usepackage{bm}
\usepackage{comment}
\usepackage{upgreek}
\usepackage[caption=false,font=footnotesize,labelfont=sf,textfont=sf]{subfig}

\usepackage{pdfpages}

\ifCLASSINFOpdf
\else
\fi
\ifCLASSOPTIONcompsoc
 \usepackage[caption=false,font=normalsize,labelfont=sf,textfont=sf]{subfig}
\else
 \usepackage[caption=false,font=footnotesize]{subfig}
\fi
\begin{document}

\includepdf[fitpaper=true, pages=-]{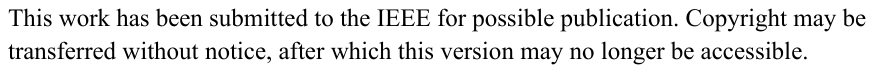}

\title{Decentralized Control Synthesis in IBR-Dominated Power Systems: A Block Diagonal Dominance Based Approach}


\author{
    Pudong~Ge,~\IEEEmembership{Member,~IEEE},
    Muhammad~Sharjeel~Javaid,~\IEEEmembership{Member,~IEEE},
    Xiaoyu~Tan,~\IEEEmembership{Graduate Student Member,~IEEE},
    Jianli~Gao, David~Angeli,~\IEEEmembership{Fellow,~IEEE},
    Janusz~Bialek,~\IEEEmembership{Fellow,~IEEE},
    and
    Balarko~Chaudhuri,~\IEEEmembership{Fellow,~IEEE}

    \thanks{The authors are with the Department of Electrical and Electronic Engineering, Imperial College London, London, SW7 2AZ, U.K. (e-mails:\{pudong.ge19, m.javaid19, xiaoyu.tan22, jianli.gao18, d.angeli, j.bialek, b.chaudhuri\}@imperial.ac.uk).}
}

\maketitle

\begin{abstract}
Integrating inverter-based resources (IBRs) from multiple vendors into power systems is challenging because their controllers are typically designed independently, with little coordination. To address this problem, this paper proposes a decentralized control synthesis framework for IBRs based on block-diagonal dominance (BDD) theory, which ensures the small-signal stability of multi-IBR power systems. By leveraging the grid frequency response, this approach enables a decentralized multi-input multi-output (MIMO) control design. Furthermore, the BDD-constrained design integrates a guaranteed minimum decay rate and defines a novel numerical metric to quantify the conservatism inherent in the decentralized stability certificate. The proposed control design with BDD constraints and minimum decay rate implementation is validated through a case study of the IEEE 9-bus test system.
\end{abstract}


%
\IEEEpeerreviewmaketitle

\section{Introduction}
%
%
%
%
\IEEEPARstart{I}{nverter}-based resources (IBRs) in power systems have changed the underlying stability characteristics, shifting the system from a predominantly synchronous paradigm to a non-synchronous one \cite{GPST_RAG,ProcIEEE_YG}. Unexpected dynamical interactions of IBRs can threaten grid stability and have motivated the extension of traditional stability definitions to encompass inverter-driven stability \cite{TPWRS2021_Uros,TPWRS_StabilityDefinition,Gao2026InSitu}. Poorly damped oscillations have already been observed in real-world IBR-dominated power systems \cite{TPWRS2023_RealWorldOscillation}. Hence, revisiting and redesigning IBR dynamics to enhance grid stability could be essential to mitigate undesired interactions among IBRs and the whole grid. IBRs can be classified into grid-forming (GFM) inverters and grid-following (GFL) ones, and both dynamic behaviors are determined by three control loops: power control, voltage control and current control \cite{TPWRS2021_Uros}. In terms of the control structure, GFM and GFL inverters share a generalized control framework in which power and voltage regulation can be integrated into a multi-input multi-output (MIMO) dynamic controller. Such a MIMO structure increases the controller's degrees of freedom (e.g., in loop configuration and controller order) and enhances its capability for grid stabilization. However, several practical challenges remain from the vendors’ perspective: system model and sufficient stability condition.

Vendors typically lack accurate system-level models of the grid dynamics.
Existing IBR control design methods frequently rely on equivalent grid models based on metrics like grid strength such as the short-circuit ratio (SCR) \cite{TSG2022_MC,TPEL2022_ZZ,TIE2023_TVT}. While previous studies have effectively validated their methods using equivalent grid models, scaling these assessments to multi-IBR systems presents a distinct challenge. Although grid strength indicators provide valuable foundational insights, capturing the intricate stability dynamics of multi-IBR configurations typically requires more detailed representations than equivalent-grid-based designs can provide. Therefore, rethinking IBR control design with explicit consideration of multi-IBR interactions and system-level stability is essential. 
Data-enabled predictive control (DeePC) avoids the need for an explicit model of grid dynamics, but this flexibility comes at the cost of requiring sufficiently rich input–output data \cite{leng_data-driven_2026}. Furthermore, industry practice demonstrates that applying frequency scanning techniques during compliance assessments enables IBR vendors to optimize control parameters and effectively mitigate IBR-induced instabilities \cite{NESO}. 
While such power system frequency response data can be made available to vendors, integrating the frequency scanning of the wider system directly into IBR parameter tuning remains a relatively underexplored approach.

Multi-IBR integration calls for scalable stability certificates that support decentralized controller design without excessive conservatism \cite{TPWRS_LH_GainPhase,GainPhase_LW,UnifiedDe_CF,haberle_decentralized_2026}. \textit{Passivity} is commonly adopted for this purpose because it enables modular assessment without a detailed external-grid model, provided that the required passivity conditions hold \cite{bao2007process,Gao2026Explicit}.
It relies on the principle that interconnecting two passive systems guarantees stability, therefore, an IBR designed to be passive will remain stable assuming the grid is also passive \cite{Passivity_overview_Harnefors,TSG2021_JDW}. Given that IBRs cannot exhibit passive behavior at low frequencies, adaptive passivity metrics, such as extended passivity, have been utilized to guide controller design \cite{TPEL_FC_ExtendedPassivity}. While existing applications primarily rely on passivity metrics to assess IBR-grid interactions, the aggregate wider-system response generally includes the dynamics of existing IBRs and therefore cannot be presumed passive without impractically isolating their effects. Direct controller synthesis from this aggregate response to certify stability in multi-IBR environments remains relatively underexplored.

Another stability criterion called \textit{gain and phase} has also been considered for multi-IBR system stability \cite{TPWRS_LH_GainPhase}. 
While integrating the small-gain theorem into linear matrix inequality (LMI) controller design provides a robust framework based on input-output gain \cite{TSG2020_LH}, this magnitude-centric approach inherently trades off phase information, which can lead to conservative stability margins.
Although \cite{TPWRS_LH_GainPhase,GainPhase_LW} proposed a valuable decentralized stability certificate using a mixed-gain-and-phase approach to reduce conservativeness, the direct integration of phase dynamics into control synthesis remains a complex challenge. Consequently, there is significant potential to develop refined sufficient conditions that further enhance the stability of multi-IBR systems.

{\black To summarize, frequency scanning provides a practical way to characterize IBR dynamics, but directly converting such frequency-domain data into synthesis-ready constraints remains challenging. One solution is to use decentralized stability certificates based on block diagonal dominance (BDD). Classical BDD theory extends
diagonal-dominance and Gerschgorin-type arguments to
block-partitioned matrices \cite{feingold_block_1962}, with later developments linking these certificates to interconnected system stability via overlapping decompositions and vector Lyapunov functions \cite{OHTA1985396}. Recent developments further show that related block-dominance conditions can lead to block-diagonal Lyapunov certificates and lower-dimensional LMIs, offering a scalable basis for analysis \cite{sootla_block-diagonal_2017,simpson_diagonal_2022}.  However, existing applications mainly use BDD as an analysis tool, rather than as an explicit synthesis constraint for IBR controller design.}

Motivated by this gap, this paper proposes a decentralized MIMO IBR controller synthesis method with BDD-based stability constraints. By using frequency-response data to characterize subsystem dynamics and interconnection effects, the proposed method reduces the reliance on detailed state-space system models while providing synthesis-oriented conditions for IBR parameter tuning by the vendors. The main contributions are as follows.

\begin{itemize}
    \item By \textbf{leveraging system-level frequency scanning data at the IBR connection points} rather than oversimplified equivalent grid metrics, {\black as formulated in \cref{sec:pre-problem}}, the proposed $3 \times 3$ MIMO IBR controller is aligned directly with actual multi-IBR operating conditions.
    \item Formulated via \textbf{LMI-based block-diagonal dominance (BDD) constraints}, {\black in \cref{sec:ctrl-lmi}}, a decentralized sufficient condition for IBR parameter tuning guarantees the scalable and stable integration of multi-IBR systems.
    \item An integrated \textbf{frequency-domain pole placement} technique, {\black introduced in \cref{cor:decay-rate}}, ensures a minimum decay rate within the BDD-constrained synthesis. 
    \item \textbf{A novel numerical metric is proposed to explicitly quantify the conservatism} existing in the BDD decentralized stability criteria {\black in \cref{dfn:conservative-metric}}.
\end{itemize}
The remainder of this paper is organized as follows. \cref{sec:pre} reviews BDD theory and formulates the problem. \cref{sec:ctrl} presents the proposed IBR controller synthesis. \cref{sec:result} provides simulation results, and \cref{sec:conclusion} concludes the paper.

\section{Preliminary}
\label{sec:pre}

\subsection{Notations}
$\bm{0},\bm{I}$ are the zero and identity matrices respectively with appropriate dimensions; $\bm{I}_n$ is an identity matrix of size $n\times n$. 
$\bm{A}^{H}$ is the Hermitian transpose of $\bm{A}$; $\bm{A}^{-1}$ is the inverse of matrix $\bm{A}$; if $\bm{A}$ has full row rank, its right inverse is defined as $\bm{A}^\mathrm{R}:=\bm{A}^{H}(\bm{A}\bm{A}^{H})^{-1}$; if $\bm{A}$ has full column rank, its left inverse is defined as $\bm{A}^\mathrm{L}:=(\bm{A}^{H}\bm{A})^{-1}\bm{A}^{H}$.
$\mathbb{N}^{+}$ denotes the set that contains all positive real integer.
$\mathcal{G}\#\mathcal{H}$ denotes the closed-loop system of open-loop system $\mathcal{G}$ negatively interconnected by a feedback controller $\mathcal{H}$; $\mathrm{diag}\{\cdot\}$ denotes a block-diagonal structure of a class of systems.
$[\cdot](s)$, $[\cdot](j\omega)$ and $[\cdot](-\alpha+j\omega)$ are matrix functions respectively in the Laplacian domain, along the $j\omega$ axis and along the $j\omega$ with decay rate $\alpha$. The norms hereafter are induced matrix norms defined as follows.

\begin{dfn}[Induced Matrix Norms]
\label{dfn:matrix-norm}
For $\bm{A}=[a_{ij}]\in\mathbb{C}^{m\times n}$, its induced matrix norm is
$\|\bm{A}\|_{p}=\sup_{\bm{x}\neq\bm{0}}\frac{\|\bm{A}\bm{x}\|_{p}}{\|\bm{x}\|_{p}}$. Hence,
$\|\bm{A}\|_{2}=\overline{\sigma}(\bm{A})=\sqrt{\lambda_{\max}(\bm{A}^{H}\bm{A})},\,
\|\bm{A}\|_{\infty} = \max_{1\leq i\leq m}\sum_{j=1}^{n}|a_{ij}|$, which have $\|\bm{A}\|_\infty\leq\sqrt{n}\|\bm{A}\|_2$. If $\bm{A}(j\omega)$ is frequency-wise, $\|\bm{A}\|_{p}=\sup_{\bm{x}\neq\bm{0},\omega}\frac{\|\bm{A}\bm{x}\|_{p}}{\|\bm{x}\|_{p}}$.
\end{dfn}

\subsection{BDD Condition for Decentralized Stability}
\begin{lem} \label{lem:bdd}
    Let $\bm{M}$ be a square complex valued block matrix
    \begin{align}
        \bm{M} = \begin{bmatrix}
            \bm{M}_{1,1} & \bm{M}_{1,2} & \cdots & \bm{M}_{1,N} \\
            \bm{M}_{2,1} & \bm{M}_{2,2} & \cdots & \bm{M}_{2,N} \\
            \vdots & & \ddots & \vdots \\
            \bm{M}_{N,1} & \bm{M}_{N,2} & \cdots & \bm{M}_{N,N}
        \end{bmatrix},
        \label{eq:M}
    \end{align}
    and $\bm{M}_{i,j}$ be one of matrix blocks in both the $i$th row and the $j$th column. For any $i\in\{1,2,\cdots,N\}$, we define
    \begin{align*}
        \bm{M}_{i,-i} = \begin{bmatrix} \bm{M}_{i,1} & \dots & \bm{M}_{i,i-1} & \bm{M}_{i,i+1} & \dots & \bm{M}_{i,N} \end{bmatrix}.
    \end{align*}
    Then $\bm{M}$ is non-singular if it holds that 
    \begin{align}
        \|\bm{M}_{i,i}^{-1}\bm{M}_{i,-i}\|_{\infty}<1.
        \label{eq:M_bdd}
    \end{align}
\end{lem}

A sufficient condition for IBR-dominated power systems can be derived based on \cref{lem:bdd}. Consider one power system with $n\in\mathbb{N}^{+}\,(n>1)$ IBRs, the model of IBRs and the rest of the grid are represented by two blocked diagonal matrices $\bm{Y}_{I}(s)$ and $\bm{Z}_{G}(s)$ respectively:
\begin{align*}
    \bm{Y}_{I} &= \mathrm{diag}\{\bm{Y}^{i}\}
    \\
    \bm{Z}_{G} &= 
        \left[
        \begin{array}{ccc}
        \bm{Z}^{11} & \cdots & \bm{Z}^{1n} \\
        \vdots & \ddots & \vdots \\
        \bm{Z}^{n1} & \cdots & \bm{Z}^{nn}
        \end{array}
        \right]
\end{align*}
with
\begin{align*}
\bm{Y}^{i} = \begin{bmatrix}
        Y_{dd}^{i} & Y_{dq}^{i} \\
        Y_{qd}^{i} & Y_{qq}^{i}
        \end{bmatrix},
\bm{Z}^{ij} = \begin{bmatrix}
        Z_{dd}^{ij} & Z_{dq}^{ij} \\
        Z_{qd}^{ij} & Z_{qq}^{ij}
        \end{bmatrix},\,\forall i,j\in\{1,\cdots,n\},
\end{align*}
which separates the closed-loop system into two parts in a negative feedback loop as \cref{fig:ibr_grid_feedback}, i.e., $\bm{Y}_{I}(s)\#\bm{Z}_{G}(s)$.

\begin{figure}[!htb]
    \centering
    \includegraphics[width=\linewidth]{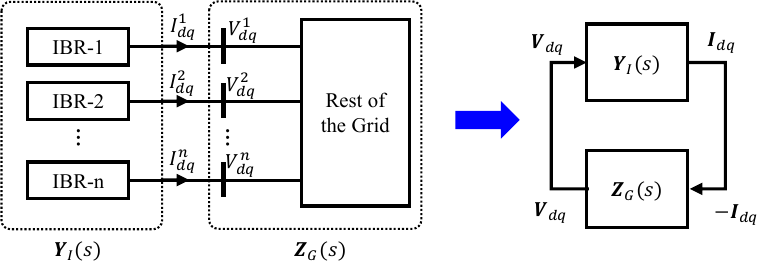}
    \caption{The power system with $n\in\mathbb{N}^{+}\,(n>1)$ IBRs, and negative feedback structure with models of IBRs and the rest of the grid.}
    \label{fig:ibr_grid_feedback}
\end{figure}

\begin{prop}[BDD Condition for Power System Stability]
    \label{prop:bdd4stability}
    Consider a power system with $n\in\mathbb{N}^{+}\,(n>1)$ IBRs and define
    $$
        \bm{L} =\bm{Y}_{I}\bm{Z}_{G}= 
        \left[
        \begin{array}{ccc}
        \bm{L}^{11} & \cdots & \bm{L}^{1n} \\
        \vdots & \ddots & \vdots \\
        \bm{L}^{n1} & \cdots & \bm{L}^{nn}
        \end{array}
        \right]
    $$ with
    $$
    \bm{L}^{ij} = \bm{Y}^{i}\bm{Z}^{ij} = \begin{bmatrix}
        L_{dd}^{ij} & L_{dq}^{ij} \\
        L_{qd}^{ij} & L_{qq}^{ij}
        \end{bmatrix},\,\forall\,i,j\in\{1,\cdots,n\},
    $$
    and
    \begin{align*}
        &\forall\,i\in\{1,\cdots,n\} ,\,\bm{M}_{i,i} = \mu\bm{I}_{2}+\bm{L}^{ii}, \\
        &\bm{M}_{i,-i} = \begin{bmatrix} \bm{L}^{i1} & \dots & \bm{L}^{i(i-1)} & \bm{L}^{i(i+1)} & \dots & \bm{L}^{in} \end{bmatrix}
    \end{align*}
    Then, the closed-loop power system $\bm{Y}_{I}(s)\#\bm{Z}_{G}(s)$ is stable if $\bm{M}(j\omega),\,\forall\,\omega\in[0,+\infty),\,\forall\,\mu\in[1,+\infty)$ is non-singular, and both $\bm{Y}_{I}(s)$ and $\bm{Z}_{G}(s)$ are internally stable.
\end{prop}

For the sake of completeness, the proofs of \cref{lem:bdd} and \cref{prop:bdd4stability} can be found in \cite{steven_bdd_arxiv}.

\subsection{Problem Formulation}\label{sec:pre-problem}

The detailed interconnected system (IBRs and rest of  grid) is outlined in \cref{fig:math_structure}. Each IBR can be modeled as $\mathcal{G}_{i}$
\begin{align}
\mathcal{G}_{i}:\left\{\begin{aligned}
    \bm{z}_{i}=\bm{G}_{i11}\bm{w}_{i}+\bm{G}_{i12}\bm{u}_{i} \\
    \bm{y}_{i}=\bm{G}_{i21}\bm{w}_{i}+\bm{G}_{i22}\bm{u}_{i}
\end{aligned}\right.,\label{eq:local-plant}
\end{align}
while the rest of the grid can be modeled as $\bm{H}(s):=\bm{Z}_{G}(s)$, which is diagonally interconnected with IBR subsystems $\mathrm{diag}\{\mathcal{G}_{i}\}$, i.e., $\mathrm{diag}\{\mathcal{G}_{i}\}\#\bm{H}(s)$. Specifically, $\bm{w}_{i}=\bm{V}_{dq}^{i}$, $
\bm{z}_{i}=\bm{I}_{dq}^{i}$, $\bm{u}_{i}=\left[I_{idr}\,I_{iqr}\,\omega_{i}\right]^{H}$ denoting IBR current control $dq$-axis reference values and IBR angular frequency, and $\bm{y}_{i}=\left[P_{im}\,V_{im}\,V_{iqm}\right]^{H}$ denoting measurement values of active power, voltage magnitude and $q$-axis voltage. For the sake of completeness, a system description is provided in \cite{ibr_model_bdd}.

\begin{figure}[!htb]
    \centering
    \includegraphics[width=\linewidth]{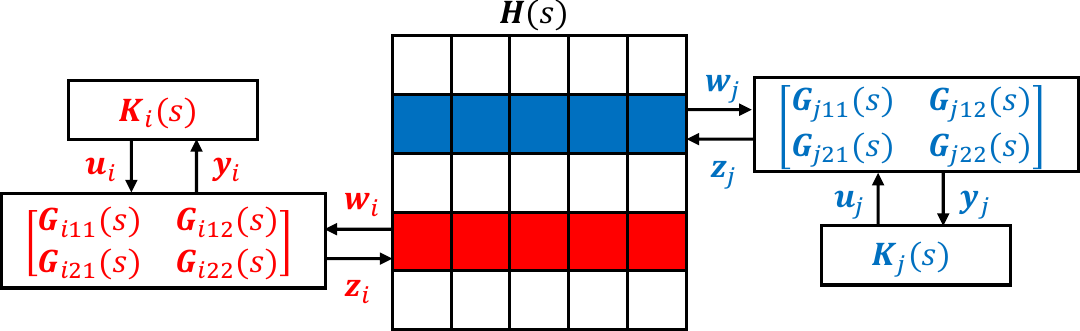}
    \caption{Detailed mathematical model of a power system with incoming IBRs: for each IBR $i$, $\bm{G}_{i11},\bm{G}_{i12},\bm{G}_{i21}$, and $\bm{G}_{i22}$ are transfer functions for open-loop local IBR model; $\bm{K}_{i}$ denotes local IBR controller; $\bm{H}$ is the frequency response for the rest of the grid looking from the IBR connection points.}
    \label{fig:math_structure}
\end{figure}

\begin{prob}[Decentralized Controller Synthesis]
    \label{prob:control}
    Consider the interconnected system depicted in \cref{fig:ibr_grid_feedback}, the objective is to synthesize a local dynamic output-feedback controller, $\bm{K}_{i}(s)$ defined by $\bm{u}_{i}=\bm{K}_{i}\bm{y}_{i}$, for the $i$-th subsystem, $\mathcal{G}_{i}$, such that the global interconnected system, denoted by $\mathrm{diag}\{\mathcal{G}_{i}\} \# \bm{H}(s)$, remains internally stable. This design is subject to a strict decentralized information constraint: the synthesis for each subsystem $i$ relies exclusively on its local model \eqref{eq:local-plant} and the provided system frequency response data $\bm{H}(j\omega)$.
\end{prob}

\begin{ass}[Pre-Integration Stability]
\label{ass:H}
    The system model represented by $\bm{H}(s)$ is internally stable at the considered operating points before incoming IBRs are connected.
\end{ass}

The system operator provides IBR vendors with the frequency response $\bm{H}(j\omega)$ of the system model, including the network and already-connected IBRs at candidate points of incoming IBR connection. In the practical compliance workflow, these data are provided by the system operator \cite{NESO}, whereas access to complete large-scale grid models and proprietary third-party IBR controls is generally infeasible. Frequency scans thus offer a practical aggregate small-signal representation without requiring disclosure of detailed white-box models.

{\black
\begin{rmk}[Vendor-Side Synthesis and Compliance Assessment]
    The same operator-provided frequency response data are used for the synthesis in \cref{prob:control} and the subsequent compliance assessment \cite{steven_bdd_arxiv}. The vendor designs or retunes its local controller to satisfy the BDD-based certificate that will be introduced later, thereby incorporating the compliance condition into the design stage rather than relying on repeated post-design testing and manual tuning.
\end{rmk}
}

\section{BDD-Constrained Decentralized IBR Control}
\label{sec:ctrl}

In this section, the BDD condition for power system stability, i.e., \cref{prop:bdd4stability} will be converted into solvable LMIs for decentralized IBR controller synthesis, where a minimum decay rate, seen as a frequency-domain regional pole placement is further integrated. Beyond this, we analyze conservativeness of the proposed decentralized IBR controller synthesis via defining a numerical metric.


\subsection{Controller Synthesis}\label{sec:ctrl-lmi}

The key idea is to utilize BDD condition in \cref{prop:bdd4stability}, as a decentralized stability criterion for controller synthesis of $\bm{K}_{i}$. 
In a block-diagonally interconnected system $\mathrm{diag}\{\mathcal{G}_{i}\}\#\bm{H}$, the grid response $\bm{H}$ can be split into its local diagonal block $\bm{H}_{i,i}$ and off-diagonal blocks
$$
\bm{H}_{i,-i} = \begin{bmatrix} \bm{H}_{i,1} & \dots & \bm{H}_{i,i-1} & \bm{H}_{i,i+1} & \dots & \bm{H}_{i,n} \end{bmatrix}.
$$

\begin{prop}[Decentralized Stability Condition for Controller Synthesis]
    \label{prop:aug-bdd}
    Consider an interconnected system $\mathrm{diag}\{\mathcal{G}_{i}\}\#\bm{H}(s)$, through \cref{ass:H} and defining $\bm{T}_{i}$ as $\bm{z}_{i}=\bm{T}_{i}\bm{w}_{i}$, controllers $\bm{K}_{i}$ stabilize the whole system if local subsystems $\mathcal{G}_{i}$ are internally stabilized by $\bm{K}_{i}$ and the following condition satisfies.
    \begin{align}
        \begin{aligned}
            &\forall\,\omega\in[0,+\infty),\,\forall\,\mu\in[1,+\infty),\\
            &\quad\|\left(\mu\bm{I}_{2}+\bm{T}_{i}\bm{H}_{i,i}\right)^{-1}\bm{T}_{i}\bm{H}_{i,-i}\|_{\infty}<1
        \end{aligned}\label{eq:bdd4thm1}
    \end{align}
\end{prop}
\begin{IEEEproof}
    The proof of \cref{prop:aug-bdd} can be achieved based on \cref{prop:bdd4stability} by imposing $\bm{T}_{i}=\bm{Y}^{i}$ and $\bm{H}_{i,j}=\bm{Z}^{ij}$.
\end{IEEEproof}

To obtain \eqref{eq:bdd4thm1} for further controller synthesis, we introduce 
\begin{align*}
    &\bm{w}_{i} = \frac{1}{\mu}\left(\bm{H}_{i,-i}\bm{d}_{i}-\bm{H}_{i,i}\bm{z}_{i}\right) 
    \\
    \implies\quad &\bm{z}_{i}=\left(\mu\bm{I}_{2}+\bm{T}_{i}\bm{H}_{i,i}\right)^{-1}\bm{T}_{i}\bm{H}_{i,-i}\bm{d}_{i}.
\end{align*}
Then we have $\bm{z}_{i}=\hat{\bm{T}}_{i}\bm{d}_{i}$ with an augmented system $\hat{\mathcal{G}}_{i}$
\begin{align}
    \hat{\mathcal{G}}_{i}:\left\{\begin{aligned}
        \bm{z}_{i}=\hat{\bm{G}}_{i11}\bm{d}_{i}+\hat{\bm{G}}_{i12}\bm{u}_{i} \\
        \bm{y}_{i}=\hat{\bm{G}}_{i21}\bm{d}_{i}+\hat{\bm{G}}_{i22}\bm{u}_{i}
    \end{aligned}\right.,\label{eq:aug-plant}
\end{align}
where $\bm{u}_{i}=\bm{K}_{i}\bm{y}_{i}$ and
\begin{align*}
    &\hat{\bm{G}}_{i11} = \left(\mu\bm{I}_{2} + \bm{G}_{i11}\bm{H}_{i,i}\right)^{-1}\bm{G}_{i11}\bm{H}_{i,-i}, 
    \\
    &\hat{\bm{G}}_{i12} = \mu\left(\mu\bm{I}_{2} + \bm{G}_{i11}\bm{H}_{i,i}\right)^{-1}\bm{G}_{i12}, 
    \\
    &\hat{\bm{G}}_{i21} = \frac{1}{\mu}\left(\bm{G}_{i21}\bm{H}_{i,-i} - \bm{G}_{i21}\bm{H}_{i,i}\hat{\bm{G}}_{i11}\right), 
    \\
    &\hat{\bm{G}}_{i22} = \bm{G}_{i22} - \frac{1}{\mu}\bm{G}_{i21}\bm{H}_{i,i}\hat{\bm{G}}_{i12}.
\end{align*}
Therefore, \eqref{eq:bdd4thm1} becomes
\begin{align}
    \label{eq:bdd4aug-plant}
    \forall\,\omega\in[0,+\infty),\,\forall\,\mu\in[1,+\infty), \|\hat{\bm{T}}_{i}\|_{\infty} < 1.
\end{align}
Inspired by \cite{SCHUCHERT2024111398}, we have the following main results for IBR control design.
\begin{thm}[Constraints for Controller Synthesis]
\label{thm:constraints}
    Consider a power system with $n$ IBRs as \cref{fig:math_structure} with \cref{ass:H}, based on the systems $\mathcal{G}_{i}$ and $\hat{\mathcal{G}}_{i}$, both IBRs and the interconnected system are stabilized if $\exists\gamma>0$, $\forall\,\omega\in[0,+\infty)$, $\forall\,\mu\in[1,+\infty)$, $\forall\,i\in\{1,\cdots,n\}$,
    \begin{subequations}
    \begin{align}
        &\begin{bmatrix}
        \frac{1}{n_t}\bm{I} - \hat{\bm{\Lambda}}_{i} & \hat{\bm{G}}_{i11} \hat{\bm{\Phi}}_{i} + \hat{\bm{G}}_{i12}\bm{X}_{i} \\
        \star & \hat{\bm{\Phi}}_{i}^{H}\hat{\bm{\Phi}}_{i}
        \end{bmatrix}(j\omega) > 0,
        \label{eq:stability-bdd}
        \\
        &\begin{bmatrix}
        \gamma\bm{I} - \bm{\Lambda}_{i} & \bm{G}_{i11} \bm{\Phi}_{i} + \bm{G}_{i12}\bm{X}_{i} \\
        \star & \bm{\Phi}_{i}^{H}\bm{\Phi}_{i}
        \end{bmatrix}(j\omega)\geq 0,
        \label{eq:stability-local}
    \end{align}
    \label{eq:stability}
    \end{subequations}
    where $n_t$ denotes the column number of $\bm{H}_{i,-i}$, and 
    \begin{align*}
        \bm{K}_{i} &= \bm{X}_{i}\bm{Y}_{i}^{-1}, \\
        \hat{\bm{\Phi}}_{i} &= \hat{\bm{G}}_{i21}^R        \left(\bm{Y}_{i} - \hat{\bm{G}}_{i22}\bm{X}_{i}\right),\, 
        \bm{\Phi}_{i} = \bm{G}_{i21}^R 
        \left(\bm{Y}_{i} - \bm{G}_{i22}\bm{X}_{i}\right), \\
        \hat{\bm{\Psi}}_{i} &= \bm{I} - \hat{\bm{G}}_{i21}^R\hat{\bm{G}}_{i21},\,
        \bm{\Psi}_{i} = \bm{I} - \bm{G}_{i21}^R\bm{G}_{i21}, \\
        \hat{\bm{\Lambda}}_{i} &= \left(\hat{\bm{G}}_{i11}\hat{\bm{\Psi}}_{i}\right) \left(\hat{\bm{G}}_{i11}\hat{\bm{\Psi}}_{i}\right)^{H},\,
        \bm{\Lambda}_{i} = \left(\bm{G}_{i11}\bm{\Psi}_{i}\right) \left(\bm{G}_{i11}\bm{\Psi}_{i}\right)^{H}.
    \end{align*}
\end{thm}
\begin{IEEEproof}
    The proof will include three parts according to \cref{prop:bdd4stability}. Firstly, $\bm{H}(s)$ is internally stable based on \cref{ass:H}.

    Secondly, we prove the interconnected stability based on \eqref{eq:stability-bdd}, i.e., \eqref{eq:stability-bdd}$\implies$\eqref{eq:bdd4aug-plant}. With $\bm{K}_{i}=\bm{X}_{i}\bm{Y}_{i}^{-1}$, we have
    \begin{align}
    \begin{aligned}
        \hat{\bm{T}}_{i} &= \hat{\bm{G}}_{i11} +\hat{\bm{G}}_{i12}\bm{X}_{i}(\bm{Y}_{i}-\hat{\bm{G}}_{i22}\bm{X}_{i})^{-1}\hat{\bm{G}}_{i21}
        \\
        &= \hat{\bm{G}}_{i11} +\hat{\bm{G}}_{i12}\bm{X}_{i}\hat{\bm{\Phi}}_{i}^{L} 
        \\
        &= \left(\hat{\bm{G}}_{i11}\hat{\bm{\Phi}}_{i}+\hat{\bm{G}}_{i12}\bm{X}_{i}\right)\hat{\bm{\Phi}}_{i}^{L} + \hat{\bm{G}}_{i11}\hat{\bm{\Psi}}_{i}
    \end{aligned},
    \label{eq:RF1-1}
    \end{align}
    \begin{align}
    \begin{aligned}
        &\implies\hat{\bm{T}}_{i}\hat{\bm{T}}_{i}^{H}= \\
        &\left(\hat{\bm{G}}_{i11}\hat{\bm{\Phi}}_{i}+\hat{\bm{G}}_{i12}\bm{X}_{i}\right)\left(\hat{\bm{\Phi}}_{i}^{H}\hat{\bm{\Phi}}_{i}\right)^{L}\left(\hat{\bm{G}}_{i11}\hat{\bm{\Phi}}_{i}+\hat{\bm{G}}_{i12}\bm{X}_{i}\right)^{H}
        \\
        &\quad+ \left(\hat{\bm{G}}_{i11}\hat{\bm{\Psi}}_{i}\right) \left(\hat{\bm{G}}_{i11}\hat{\bm{\Psi}}_{i}\right)^{H}\leq\|\hat{\bm{T}}_{i}\|_{2}^{2}\cdot\bm{I}.
    \end{aligned}
    \label{eq:RF1-2}
    \end{align}
    Using matrix norm inequality in \cref{dfn:matrix-norm} and Schur Complement, there exists 
    \begin{align}
        \begin{aligned}
        & \|\hat{\bm{T}}_{i}\|_{\infty}^{2}\leq{n_t}\|\hat{\bm{T}}_{i}\|_{2}^{2}<1 
        \iff \\
        &\quad\begin{bmatrix}
        \frac{1}{n_t}\bm{I} - \hat{\bm{\Lambda}}_{i} & \hat{\bm{G}}_{i11} \hat{\bm{\Phi}}_{i} + \hat{\bm{G}}_{i12}\bm{X}_{i} \\
        \star & \hat{\bm{\Phi}}_{i}^{H}\hat{\bm{\Phi}}_{i}
        \end{bmatrix} > 0.
        \end{aligned}
        \label{eq:RF1-3}
    \end{align}

    Thirdly, followed by a similar procedure \eqref{eq:RF1-1} -- \eqref{eq:RF1-3} to local subsystems \eqref{eq:local-plant}, \eqref{eq:stability-local} implies local stability.
\end{IEEEproof}

\begin{rmk}[Solution to Model Rank Deficit]
    For the system $\hat{\mathcal{G}}_{i}$, $\mathrm{rank}\left(\hat{\bm{G}}_{i21}\right)=2$ while $\bm{y}_{i}$ having 3 outputs. To satisfy the condition $\forall\,\omega\in[0,+\infty),\hat{\bm{G}}_{i21}(j\omega)$ having full row rank for its right inverse $\hat{\bm{G}}_{i21}^{R}$ in \eqref{eq:RF1-1} as \cite{SCHUCHERT2024111398}, a regularization process
    \begin{align*}
        &\bm{z}_{i} = \begin{bmatrix} \hat{\bm{G}}_{i11} & \bm{0} \end{bmatrix} \begin{bmatrix} \bm{w}_{i} \\ \bm{v}_{i} \end{bmatrix} + \hat{\bm{G}}_{i12}\bm{u}_{i}, \\
        &\bm{y}_{i} = \begin{bmatrix} \hat{\bm{G}}_{i21} & \epsilon_{y}\bm{I}_{3} \end{bmatrix} \begin{bmatrix} \bm{w}_{i} \\ \bm{v}_{i} \end{bmatrix} + \hat{\bm{G}}_{i22}\bm{u}_{i}
    \end{align*}
     is applied with a small constant $\epsilon_{y}>0$. This regularization process should also be applied to the IBR system $\mathcal{G}_{i}$.
\end{rmk}

\begin{rmk}[Controller Structure]
    In \cref{thm:constraints}, the controller $\bm{K}_{i}=\bm{X}_{i}\bm{Y}_{i}^{-1}$ can be designed as matrix polynomials through parametrization as \cite{SCHUCHERT2024111398}, which releases full degrees of freedom on controller structure and order. In other words, $\bm{K}_{i}$, a $3\times3$ MIMO controller can be assigned in any structures with predefined orders. {\black Thanks to this flexibility, appropriate weighting filters allow the proposed approach to specify input–output control bandwidths and enforce desired GFL- or GFM-like IBR dynamic characteristics \cite{TSG2020_LH}. In this paper, the proposed method does not prescribe GFL- or GFM-like IBR dynamic characteristics. Instead, it shapes MIMO controller dynamics for improved performance under BDD-based stability constraints.}
\end{rmk}

Constraints in \cref{thm:constraints} cannot be solved owing to non-convex terms $\hat{\bm{\Phi}}_{i}^{H}\hat{\bm{\Phi}}_{i},\bm{\Phi}_{i}^{H}\bm{\Phi}_{i}$, which can be convexification through Taylor expansion, i.e.,
\begin{align*}
    \hat{\bm{\Phi}}_{i}^{H}\hat{\bm{\Phi}}_{i} \geq \hat{\bm{\Phi}}_{i}^{H}\hat{\bm{\Phi}}_{ic}+\hat{\bm{\Phi}}_{ic}^{H}\hat{\bm{\Phi}}_{i}-\hat{\bm{\Phi}}_{ic}^{H}\hat{\bm{\Phi}}_{ic}, \\
    \bm{\Phi}_{i}^{H}\bm{\Phi}_{i} \geq \bm{\Phi}_{ic}^{H}\bm{\Phi}_{i}+\bm{\Phi}_{i}^{H}\bm{\Phi}_{ic}-\bm{\Phi}_{ic}^{H}\bm{\Phi}_{ic},
\end{align*}
where $\hat{\bm{\Phi}}_{ic},\bm{\Phi}_{ic}$ are arbitrary known matrices via $$\hat{\bm{\Phi}}_{ic} = \hat{\bm{G}}_{i21}^R\left(\bm{Y}_{ic}-\hat{\bm{G}}_{i22}\bm{X}_{ic}\right),\,\bm{\Phi}_{ic} = \bm{G}_{i21}^R\left(\bm{Y}_{ic} - \bm{G}_{i22}\bm{X}_{ic}\right).$$ Then \eqref{eq:stability} is linearized into
\begin{subequations}
\begin{align}
    &\begin{bmatrix}
    \frac{1}{n_t}\bm{I} - \hat{\bm{\Lambda}}_{i} & \hat{\bm{G}}_{i11} \hat{\bm{\Phi}}_{i} + \hat{\bm{G}}_{i12}\bm{X}_{i} \\
    \star & \hat{\bm{\Phi}}_{i}^{H}\hat{\bm{\Phi}}_{ic}+\hat{\bm{\Phi}}_{ic}^{H}\hat{\bm{\Phi}}_{i}-\hat{\bm{\Phi}}_{ic}^{H}\hat{\bm{\Phi}}_{ic}
    \end{bmatrix}(j\omega) > 0,
    \label{eq:stability-bdd2}
    \\
    &\begin{bmatrix}
    \gamma\bm{I} - \bm{\Lambda}_{i} & \bm{G}_{i11} \bm{\Phi}_{i} + \bm{G}_{i12}\bm{X}_{i} \\
    \star & \bm{\Phi}_{ic}^{H}\bm{\Phi}_{i}+\bm{\Phi}_{i}^{H}\bm{\Phi}_{ic}-\bm{\Phi}_{ic}^{H}\bm{\Phi}_{ic}
    \end{bmatrix}(j\omega)\geq 0.
    \label{eq:stability-local2}
\end{align}
\label{eq:stability2}
\end{subequations}
Thanks to the stability analysis through Nyquist plot in \cite{SCHUCHERT2024111398}, it implies interconnected stability and IBR stability are still satisfied with \eqref{eq:stability-bdd2} and \eqref{eq:stability-local2} respectively if the known controller $\bm{K}_{ic}$ is a locally stabilizing controller.

\begin{rmk}[Practical Implementation of Controller Synthesis]
    Both sufficient conditions in \cref{prop:aug-bdd} and \cref{thm:constraints} should be considered for $\forall\,\omega\in[0,+\infty)$ and $\forall\,\mu\in[1,+\infty)$ for each IBR system. A common approach is to sample these infinite number of LMI constraints for all points in $\Omega=\{\omega_{1},\cdots,\omega_{k}\}\subset[0,+\infty)$ and $\Upsilon\in\{\mu_{1},\cdots,\mu_{t}\}\subset[1,+\infty)$, leading to a large number, i.e., $(k\cdot t)$ BDD-based LMI constraints that can be handled via numerical semi-definite programming solvers. The overall procedure for the proposed controller synthesis is detailed in \cref{alg:1}.
\end{rmk}

\begin{algorithm}[!htb]
    \caption{BDD-based decentralized synthesis problem implementation of IBR MIMO control}\label{alg:1}
    \KwData{Frequency spectral $\bm{Z}_{G}(s)$ at incoming IBR connection points from System Operators; system model \eqref{eq:local-plant}; initial stabilized controllers $\bm{K}_{ic}=\bm{X}_{ic}\bm{Y}_{ic}^{-1}$; frequency samples $\Omega=\{\omega_{1},\cdots,\omega_{k}\}$; BDD constant samples $\Upsilon\in\{\mu_{1},\cdots,\mu_{t}\}$.}
    \KwResult{Controllers $\bm{K}_{i}$ for IBR Vendors.}
    \BlankLine
    \ForEach{IBR $i\in\{1,\cdots,n\}$}{
        \Repeat{$|\gamma^{\star}-\gamma|\leq\varepsilon$}{
        \tcc{Solve the sampled-frequency version of the optimization  problem; initialization by $\gamma=+\infty$, pre-defined convergence tolerance $\varepsilon$}
        Renew $\gamma^{\star} = \gamma$ \;
        Optimize to update $\gamma$
        \begin{align}
        \begin{aligned}
            &\min_{\bm{X}_{i},\bm{Y}_{i}}\gamma \\
            &\forall\,\omega\in\Omega,\,\forall\,\mu\in\Upsilon,\,\eqref{eq:stability2}
        \end{aligned}
        \label{eq:stability-alg:1}
        \end{align}
        }
    Obtain controller via $\bm{K}_{i}=\bm{X}_{i}\bm{Y}_{i}^{-1}$\;
    }
\end{algorithm}

\subsection{Conservativeness Analysis and Minimum Decay Rate}\label{sec:ctrl-conservative-decay}


In \cref{thm:constraints}, \eqref{eq:stability-bdd} and \eqref{eq:stability-local} ensure BDD-based decentralized stability and local IBR stability respectively. \eqref{eq:stability-bdd} brings inevitable conservativeness, hence, there exists interconnected stability where \eqref{eq:stability-bdd3} satisfies with BDD constant $\kappa>1$, i.e., $\|\left(\mu\bm{I}_{2}+\bm{T}_{i}\bm{H}_{i,i}\right)^{-1}\bm{T}_{i}\bm{H}_{i,-i}\|_{\infty}<\kappa$.
\begin{align}
\begin{aligned}
    & \forall\,\omega\in[0,+\infty), \forall\,\mu\in[1,+\infty), \exists\kappa>1, \\
    &\quad\begin{bmatrix}
    \frac{1}{n_t}\kappa^2\bm{I} - \hat{\bm{\Lambda}}_{i} & \hat{\bm{G}}_{i11} \hat{\bm{\Phi}}_{i} + \hat{\bm{G}}_{i12}\bm{X}_{i} \\
    \star & \hat{\bm{\Phi}}_{i}^{H}\hat{\bm{\Phi}}_{i}
    \end{bmatrix}(j\omega) > 0.
\end{aligned}
\label{eq:stability-bdd3}
\end{align}

\begin{dfn}[Conservativeness Metric on BDD-Constrained Decentralized Stability]
\label{dfn:conservative-metric}
    Consider a power system with $n$ IBRs as \cref{fig:math_structure}, $\kappa_{-1}^{\star}=\kappa^{\star}-1$ is defined as one metric quantifying the BDD-constrained decentralized stability conservativeness and
    \begin{align*}
    \begin{aligned}
        &\kappa^{\star}=\arg\max_{\kappa}  \\ &\quad\Big(\mbox{Run \cref{alg:1}} \enspace \mathrm{where} \enspace \eqref{eq:stability-bdd2} \enspace \mbox{is replaced with} \enspace \eqref{eq:stability-bdd3} \\ 
        &\qquad\qquad\qquad\qquad\qquad\mathrm{stabilises} \enspace \mathrm{diag}\{\mathcal{G}_{i}\}\#\bm{H}(s)\Big)
    \end{aligned}
    \end{align*}
    which can be obtained through \cref{alg:2}.
    
    \begin{algorithm}[!htb]
        \caption{Quantifying conservativeness metric on the BDD decentralized stability certificate}\label{alg:2}
        \KwData{Data in \cref{alg:1}; $\kappa=1$; conservativeness incremental precision $\varepsilon_{\kappa}$.}
        \KwResult{Conservative metric $\kappa_{-1}^{\star}$.}
        \BlankLine
        \Repeat{System $\mathrm{diag}\{\mathcal{G}_{i}\}\#\bm{H}(s)$ destabilized}{
            Run \cref{alg:1} where \eqref{eq:stability-bdd2} is replaced with \eqref{eq:stability-bdd3} in \eqref{eq:stability-alg:1}\;
            $\kappa_{-1}^{\star}=\kappa-1$\;
            $\kappa=\kappa+\varepsilon_{\kappa}$\;
        }
    \end{algorithm}
    \end{dfn}

In addition, imposing a minimum decay rate strictly ensures the system can meet strict system performance beyond stabilization, e.g., voltage and power tracking damping.
\begin{cor}[Minimum Decay Rate for Controller Synthesis]
\label{cor:decay-rate}
    Based on the systems in \cref{thm:constraints}, both IBRs and the interconnected system are stabilized if $\exists\gamma>0$, $\forall\,\omega\in[0,+\infty)$, $\forall\,\mu\in[1,+\infty)$, $\forall\,i\in\{1,\cdots,n\}$, \eqref{eq:stability-decay} satisfies.
    \begin{subequations}
    \begin{align}
        \begin{bmatrix}
        \frac{1}{n_t}\bm{I} - \hat{\bm{\Lambda}}_{i} & \hat{\bm{G}}_{i11} \hat{\bm{\Phi}}_{i} + \hat{\bm{G}}_{i12}\bm{X}_{i} \\
        \star & \hat{\bm{\Phi}}_{i}^{H}\hat{\bm{\Phi}}_{i}
        \end{bmatrix}(-\alpha+j\omega) > 0
        \label{eq:stability-bdd4}
        \\
        \begin{bmatrix}
        \gamma\bm{I} - \bm{\Lambda}_{i} & \bm{G}_{i11} \bm{\Phi}_{i} + \bm{G}_{i12}\bm{X}_{i} \\
        \star & \bm{\Phi}_{i}^{H}\bm{\Phi}_{i}
        \end{bmatrix}(-\alpha+j\omega)\geq 0
        \label{eq:stability-local4}
    \end{align}
    \label{eq:stability-decay}
    \end{subequations}
    where $\alpha\geq0$ denotes a desired minimum decay rate.
\end{cor}
In \cref{cor:decay-rate}, through replacing $j\omega$ with $-\alpha+j\omega$, we shift both open-loop models and parametrized controllers, which enforces a minimum decay rate $\alpha$ by executing the entire controller synthesis in a shifted Laplace domain and further forces interconnected closed-loop poles to the left of $-\alpha$ in all possible oscillatory frequencies. Hence, \cref{cor:decay-rate} has a similar proof to \cref{thm:constraints}.
Similar to \cref{thm:constraints}, \eqref{eq:stability-decay} can be converted into LMIs as \eqref{eq:stability2}. Then \cref{cor:decay-rate} can be solved iteratively and accordingly as \cref{alg:1}, and a conservativeness analysis specialized for \cref{cor:decay-rate} can be achieved similarly as \cref{alg:2}.

\begin{rmk}[Practical Feasibility of \cref{cor:decay-rate}]
    Since frequency scan is available to IBR vendors, imposing minimum decay rate as \cref{cor:decay-rate} may face a practical obstacle. If the frequency scan is known via parametric transfer functions, the complex contour shift is trivial. However, if frequency responses are purely data-driven, such as bode plot data points, measured strictly along the $j\omega$ axis, it is infeasible to multiply a $j\omega$ data point by a scalar to accurately find the $-\alpha+j\omega$ data point without knowing the internal dynamics. To solve this problem, a modified physical testing (i.e., inject exponential signals ($e^{-\alpha t}\sin(\omega t)$) to directly measure the frequency response along the $-\alpha+j\omega$ contour) or a vector fitting to identify a parametric transfer function can be applied.
\end{rmk}

\section{Results}
\label{sec:result}
In this section, a modified IEEE 9-bus system \cite{stability_book} with 3 IBRs is considered to validate the proposed BDD-constrained decentralized IBR MIMO control design. \cref{fig:topology} shows the network topology. System is modeled in EMT-$dq$ domain, as described in \cite{EMTdq_Javaid}, and all relevant parameters are given in \cite{ibr_model_bdd}.

\begin{figure}[!htb]
    \centering
    \includegraphics[width=\linewidth]{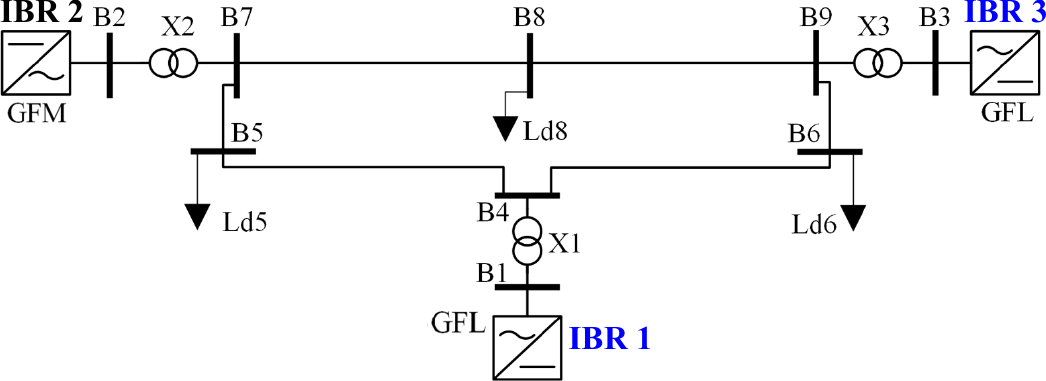}
    \caption{Modified IEEE 9-bus system: IBR 1 and IBR 3 are designed with the proposed BDD decentralized stability constraints.}
    \label{fig:topology}
\end{figure} 

\subsection{Control Performance via Decentralized Synthesis}


The proposed decentralized control synthesis is utilized for IBR 1 and IBR 3. \cref{fig:sim1}a shows a range of system eigenvalues and \cref{fig:sim1}b shows local BDD index tests. The strict placement of eigenvalues in the left half-plane corresponds with both IBRs successfully passing the BDD test, which is further illustrated in power and voltage tracking performance in \cref{fig:sim1}c and \cref{fig:sim1}d. The BDD index in \cref{fig:sim1}b, evaluated via the infinity norm $\|\hat{\bm{T}}_i\|_\infty$, serves as a metric for decentralized stability in interconnected systems. It is shown that for a frequency sweep between $1-1000\,\mathrm{Hz}$, the infinity norm stays strictly below the critical threshold of $1$. The time-domain responses in \cref{fig:sim1}c and \cref{fig:sim1}d further confirmed consistent performance under simultaneous $0.01\,\mathrm{pu}$ step changes in the active-power and voltage references at $t=0.1\,\mathrm{seconds}$.

\begin{figure}[!htb]
    \centering
    \includegraphics[width=\linewidth]{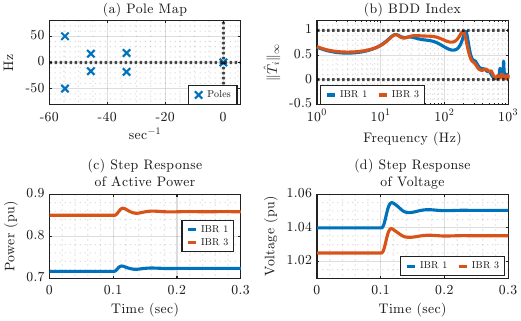}
    \caption{Validations of the proposed BDD-constrained decentralized IBR controller synthesis.}
    \label{fig:sim1}
\end{figure}

\subsection{Minimum Decay Rate Comparison}

Integrating different minimum decay rates into the proposed controller synthesis yields frequency-domain comparisons shown in \cref{fig:sim2}, demonstrating a trade-off between small-signal stability margins, robustness, and closed-loop bandwidth.
As shown in \cref{fig:sim2}a, changing $\alpha$ shifts the pole locations, but the dominant poles remain on the stable side of the complex plane for all tested values. The BDD index curves in \cref{fig:sim2}b and \cref{fig:sim2}c remain below the unit threshold for both IBRs, which suggests that the interconnection effects remain sufficiently dominated by the local subsystem dynamics under the considered decay settings. This confirms that decentralized synthesis ensures the interconnected-system stability.

\begin{figure}[!htb]
    \centering
    \includegraphics[width=\linewidth]{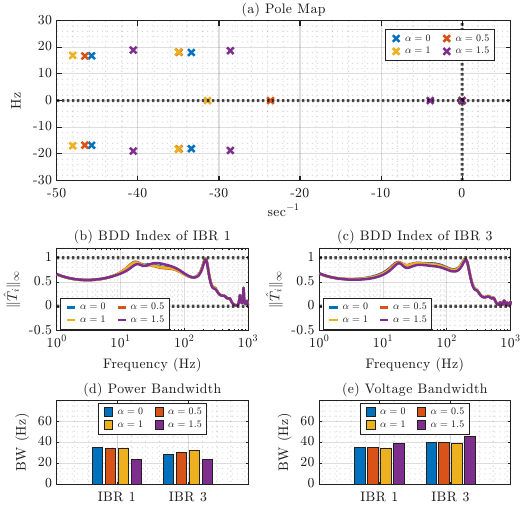}
    \caption{Imposing minimum decay rate in the proposed BDD-constrained decentralized IBR controller synthesis.}
    \label{fig:sim2}
\end{figure}

The bandwidth analyses in \cref{fig:sim2}d and \cref{fig:sim2}e show that $\alpha$ changes the active-power and voltage bandwidths in different ways. This implies that the decay parameter does not simply increase or decrease the overall dynamic speed, but redistributes the dynamic characteristics between power and voltage channels. Overall, the case of imposing minimum decay rate suggests that $\alpha$ can be used to tune transient performance while maintaining stability.

\subsection{Conservativeness Analysis}


The conservatism in the BDD-based decentralized stability certificate has been investigated by systematically relaxing the BDD constraint through the scaling constant $\kappa$ in \cref{fig:sim3}. This analysis reveals the gap between the BDD-based theoretical stability boundary and the system stability limit of the interconnected multi-IBR system. The dominant poles in \cref{fig:sim3}a move closer to the imaginary axis as $\kappa$ increases, showing a reduction in the damping. This trend is also reflected in \cref{fig:sim3}b and \cref{fig:sim3}c, where the BDD index peaks increase and exceed the unit threshold for larger $\kappa$ values. The system remains stable for $\kappa\in\{1.41,1.73,1.76\}$, demonstrating that the strict $\kappa=1$ requirement is mathematically conservative. It is not until $\kappa=1.79$ that the dominant poles cross into the right-half of the complex plane, indicating instability. 

\begin{figure}[!htb]
    \centering
    \includegraphics[width=\linewidth]{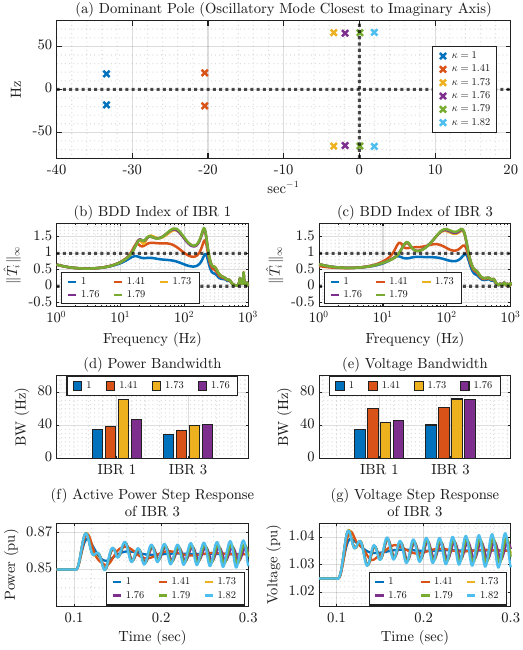}
    \vspace{-1em}
    \caption{Conservative analysis of the proposed BDD-constrained decentralized IBR controller synthesis: legends in subfigures are based on $\kappa$ values. In the subfigures (b) and (c), the case of $\kappa=1.82$ is omitted because it almost overlaps with the case of $\kappa=1.79$, while the subfigures (d) and (e) only show the stable cases.} 
    \label{fig:sim3}
\end{figure}

The bandwidth in \cref{fig:sim3}d and \cref{fig:sim3}e further show that larger $\kappa$ modifies both active-power and voltage bandwidths. The time-domain responses in \cref{fig:sim3}f and \cref{fig:sim3}g confirm the same behavior: larger $\kappa$ produces more persistent oscillations after the disturbance. These results indicate that increasing $\kappa$ relaxes the bound but reduces the damping, so $\kappa$ must be selected conservatively to avoid weakly damped responses.

\begin{figure}[!htb]
    \centering
    \includegraphics[width=\linewidth]{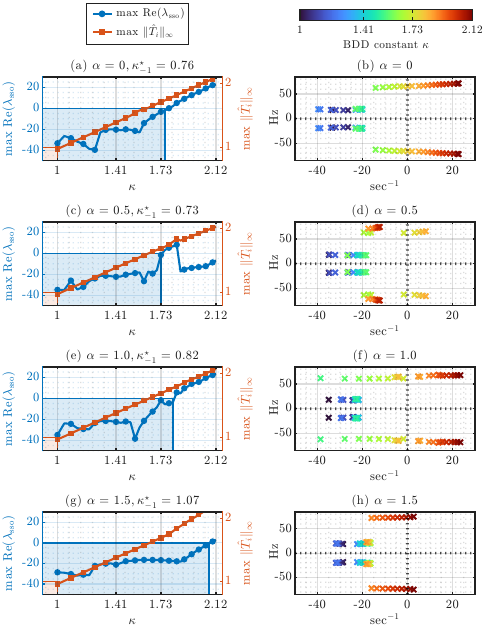}
    \vspace{-1em}
    \caption{Conservativeness analysis in terms of imposing different minimum decay rates. Left axis in blue denotes the distance of real part of least-damped sub-synchronous mode $\max\,\mathrm{Re}(\lambda_\mathrm{sso})$, while right axis in red denotes the BDD index. With increasing $\alpha$ (top-down) the conservativeness metric in blue area becomes larger and the $\kappa$ satisfying BDD-based decentralized stability in red remains unchanged. In the right column, the dominant pole changes with $\kappa$ corresponding to blue line in the left.}
    \label{fig:sim4}
\end{figure}

\cref{fig:sim4} illustrates a general, though non-monotonic, relationship between the imposed minimum decay rates $\alpha$ and the resulting conservativeness of the stability certificate, denoted by $\kappa^\star_{-1}$. The sub-figures systematically detail this relationship by combining theoretical mathematical bounds, i.e., BDD indices, against system stability.
In the left column, the red curves with corresponding right axes denote BDD indices, while the blue curves with left axes track the real part of the least-damped sub-synchronous mode, $\max \text{Re}(\lambda_{\text{sso}})$. The instability strictly occurs when this blue trajectory crosses the zero axis.
The conservativeness is captured by the blue shaded regions, which define the gap between the loss of the theoretical BDD guarantee and the onset of system instability. As $\alpha$ increases from top to bottom, the red shaded area (indicating the region satisfying strictly BDD-based decentralized stability) remains unchanged. The conservativeness metric, represented by the blue shaded area, exhibits an overall expanding trend despite a slight initial contraction. Specifically, as $\alpha$ increases, the allowable relaxation $\kappa^\star_{-1}$ before reaching instability dips slightly, before growing substantially to $0.82$ at $\alpha=1.0$ and $1.07$ at $\alpha=1.5$.
The right-column dominant poles support this interpretation by showing that higher $\alpha$  values shift the pole clusters toward the imaginary axis and, in some cases, toward the right half-plane. Interestingly, the dominant poles for all cases are initialized at a relatively similar depth within the left-half of the complex plane (near $-20\,\text{sec}^{-1}$ for $\kappa \approx 1$, indicated by dark blue markers). 
However, as $\alpha$ increases, particularly to $\alpha=1.5$, the pole trajectories become less sensitive to $\kappa$, requiring a substantially larger $\kappa$ to induce instability. This behavior suggests that the stronger decay-rate constraint reshapes the synthesized controller and reduces the sensitivity of the closed-loop poles to $\kappa$, rather than uniformly shifting them farther into the left-half plane for each $\kappa$.

Overall, this analysis underscores a fundamental control design trade-off: while enforcing stricter decay rates (i.e., larger $\alpha$) improves the predefined small-signal damping, it generally makes the BDD stability bounds increasingly conservative. 
In addition, the admissible BDD constant $\kappa$ is determined by the prescribed decay rate $\alpha$ and therefore represents an $\alpha$-dependent measure of conservatism.

\section{Conclusion}
\label{sec:conclusion}

This paper has shown the application of BDD theory for decentralized IBR control design. The proposed framework ensures the small-signal stability of general multi-IBR power systems. Translating the BDD condition into LMIs for controller synthesis ensures the incoming IBRs comply with grid stability requirements. Furthermore, the implementation of a minimum decay rate ensures a well-damped response from a small-signal stability perspective. It has been demonstrated that while such a decentralized stability certificate introduces conservatism, it simultaneously provides a robust stability margin that can be systematically increased by imposing stricter minimum decay rates. A key merit of the proposed IBR controller synthesis lies in its practical feasibility as it successfully bridges the information gap between system operators and IBR vendors by utilizing system-level frequency scans, which are often the only accessible grid models. 

It should be noted that the decentralized IBR control framework developed herein establishes a sufficient but not necessary condition under BDD theory. In addition, its single-point linearization lacks explicit robustness guarantees across multiple operating points. Hence, developing a tractable synthesis method that provides stability certificates while ensuring robust performance against diverse grid conditions remains an open challenge that we are addressing in our follow-up work.

\ifCLASSOPTIONcaptionsoff
  \newpage
\fi



%
\bibliographystyle{IEEEtran}
\bibliography{ref}
%








\end{document}